\documentclass[]{interact}

\usepackage{epstopdf}
\usepackage[caption=false]{subfig}
\usepackage{xcolor}

\usepackage[ruled,vlined]{algorithm2e}
\usepackage{url}
\usepackage[numbers,sort&compress]{natbib}
\bibpunct[, ]{[}{]}{,}{n}{,}{,}
\renewcommand\bibfont{\fontsize{10}{12}\selectfont}
\theoremstyle{plain}
\newtheorem{theorem}{Theorem}[section]
\newtheorem{lemma}[theorem]{Lemma}

\theoremstyle{definition}
\newtheorem{definition}[theorem]{Definition}
\newtheorem{example}[theorem]{Example}

\theoremstyle{remark}
\newtheorem{remark}{Remark}

\usepackage{Notation}

\begin{document}
	
	
	\title{A Rosenbrock framework for tangential interpolation of port-Hamiltonian
		descriptor systems}
	
	\author{
		\name{Tim Moser\textsuperscript{$\dagger$,$\ddagger$}\thanks{\textsuperscript{$\ddagger$} Corresponding author: Tim Moser; Email:
				tim.moser@tum.de} and Boris Lohmann\textsuperscript{$\dagger$}}
		\affil{\textsuperscript{$\dagger$} 
			Technical University of Munich, TUM School of Engineering and Design,
			Department of
			Engineering Physics and Computation, Boltzmannstr. 15, 85748
			Garching, Germany}
	}
	
	\maketitle
	
	\begin{abstract}
		We present a new structure-preserving model order reduction (MOR) framework for large-scale port-Hamiltonian descriptor systems (pH-DAEs). Our method exploits the structural properties of the Rosenbrock system matrix for this system class and utilizes condensed forms which often arise in applications and reveal the solution behaviour of a system. Provided that the original system has such a form, our method produces reduced-order models (ROMs) of minimal dimension, which tangentially interpolate the original model's transfer function and are guaranteed to be again in pH-DAE form. This allows the ROM to be safely coupled with other dynamical systems when modelling large system networks, which is useful, for instance, in electric circuit simulation.
	\end{abstract}
	

	\section{Introduction}\label{sec:intro}
	
	The port-Hamiltonian (pH) modelling paradigm provides an energy-based framework for constructing high-fidelity models of complex dynamical systems. The separation between constitutive relations and the interconnection structure enables a modular modelling approach, in which different subsystems are modelled independently and then interconnected via power flows while preserving important physical properties~\cite{Duindam2009}. This flexibility is particularly advantageous when considering interactions between subsystems across different physical domains or time scales~\cite{SchJ14}. Furthermore, network laws that govern these interactions, such as Kirchhoff’s laws in electrical circuits or position and velocity constraints in mechanical systems, may be incorporated via algebraic constraints, leading to pH differential-algebraic equation systems (pH-DAEs)~\cite{MehU22}. 
	
	If the network increases in complexity, generally, the state-space dimension of its associated pH-DAE model does so too, which makes the simulation and design of model-based controllers computationally challenging. Model order reduction (MOR) may be applied to approximate parts of the network with reduced-order models (ROMs) of substantially smaller dimension to decrease the computational costs. However, this comes with two major challenges. On the one hand, the impacts of the algebraic equations on the model dynamics have to be reflected in the ROM as well, but without significantly increasing the reduced state-space dimension. On the other hand, to facilitate the subsequent coupling of the ROM with the remaining network, the MOR method should be \emph{structure-preserving}, meaning that the ROM is again a pH-DAE. Structure-preserving MOR of pH-DAEs is, therefore, an active research field; see  \cite{MosSMV22,MosSMV2022b,HauMM19,BeaGM21} for recent work. 
	
	In this work, we focus on interpolatory MOR techniques, which are particularly suited for the reduction of very large-scale models due to their computational efficiency. While these methods are well-established for various system classes, including general DAE systems (see \cite{Antoulas2020} for a comprehensive overview), this is still only partially the case for pH-DAEs. Tangential interpolation of pH ordinary differential equation systems (pH-ODEs) with no algebraic constraints has been addressed in \cite{Gugercin2009,Gugercin2012,Polyuga2009,Polyuga2010,Polyuga2011} and is well understood. Special pH-DAE cases with algebraic constraints have been considered in \cite{HauMM19,BeaGM21}. However, except for a special representative with differentiation index 2, the proposed methods either do not guarantee that the reduced model is again a pH-DAE system or only reduce the dynamical part of the system equations and retain the algebraic part, which is not the maximal possible reduction. 
	The research question we address in this work is the following: Is there a unifying framework for structure-preserving tangential interpolation that is applicable to all types of linear, time-invariant pH-DAEs? 
	
	A good starting point to answer this question is the condensed forms presented in \cite{Achleitner2021,BeaGM21}. While condensed forms are generally helpful in analyzing the solution behaviour of DAEs, they are also useful for MOR since they allow the impact of algebraic constraints on the system dynamics to be analyzed. In general, the computation of these condensed forms relies on a series of rank conditions which may be sensitive under perturbations; see e.g. \cite{Byers2007}. Fortunately, as demonstrated in \cite{Guducu2021,MehU22} for various physical examples, this can be considered directly in the modelling process such that the resulting model already has (or is close to) condensed form. Since DAEs often include redundant algebraic equations which increase the computational cost for simulation, the task of MOR is to identify a minimal set of equations that describes the dynamical behavior of the system. In 1970, Rosenbrock \cite{Rosenbrock1970} proposed to represent dynamical systems with a polynomial matrix, which is beneficial to compute minimal representations and nowadays known as the \emph{Rosenbrock system matrix}. Interestingly enough, this matrix has not extensively been exploited in the context of MOR. 
	
	In this work, we show that a combination of condensed forms and the system-theoretic concept of the Rosenbrock system matrix, which has a particular structure for pH-DAEs, can be exploited to derive a unifying interpolatory MOR framework that
	\begin{itemize}
		\item[(i)] can be used for linear, time-invariant pH-DAEs with arbitrary differentiation index,
		\item[(ii)] guarantees ROMs in pH-DAE form with a minimal state-space dimension,
		\item[(iii)] works with the original (typically sparse) system matrices and is therefore computationally efficient, and
		\item[(iv)] enables a straightforward adaptation and integration of different interpolatory MOR strategies that have been developed for general dynamical systems.
	\end{itemize}
	
	The paper is organized as follows. In Section \ref{sec:prelims}, we use the concept of the Rosenbrock system matrix to introduce the system class of pH-DAEs and its condensed forms and recapitulate the basics of tangential interpolation. In Section \ref{sec:our_approach}, we propose a new interpolatory MOR framework that exploits the structural properties of pH-DAEs in condensed form and show how this framework naturally generalizes to several extensions proposed for unstructured systems in Section \ref{sec:extensions}. We conclude with an application of our method to two electric circuits and a discussion of the above claims in Sections \ref{sec:num_examples} and \ref{sec:conclusion}.
	
	\section{Preliminaries}\label{sec:prelims}
	We consider linear time-invariant (LTI) systems of the form
	\begin{equation}
		\label{eq:ss_def}
			\begin{split}
				\pE \pdx(t) &= \pA \px(t) + \pB u(t), \quad \px(0) = 0,\\
				\py(t) &= \pC \px(t) + \pD u(t),
			\end{split}	
	\end{equation}
	with state vector $x(t) \in\R^{n}$, inputs $u(t) \in\R^{m}$, outputs $y(t)
	\in\R^{m}$ for all $t \in [0, \infty)$ and constant matrices
	$\pE,\,\pA\in\R^{n\times n}$, $\pB\in\R^{n\times m}$, $\pC\in\R^{m\times n}$ and
	$\pD\in\R^{m\times m}$. Systems with a singular descriptor matrix $\pE$ are referred to as DAE systems, and ODE systems otherwise. In the following, we will
	assume that the pencil $\lambda\pE-\pA$ is \emph{regular}, i.e. $\det(\lambda\pE-\pA)\neq 0$ for some $\lambda\in\C$. We denote the ring of polynomials with coefficients in $\R$ by $\R[s]$ and the set of $n\times m$ matrices with entries in $\R[s]$ by $\R[s]^{n\times m}$.
	
	\subsection{The Rosenbrock system matrix}
	After a Laplace transformation of the state-space equations in \eqref{eq:ss_def}, we obtain the following equations in matrix form:
	\begin{equation} \label{eq:FOM_laplace}
		\Ros(s)\mat{\xi(s)\\ \nu(s)}=\mat{0\\\psi(s)},\quad \text{where } \Ros(s) :=
		\mat{s\pE-\pA&-\pB\\\pC&\pD}\in\R[s]^{(n+m)\times (n+m)}.
	\end{equation}
	The polynomial matrix $\Ros$ is also called the \emph{Rosenbrock system matrix} (in the following: system matrix) \cite{Rosenbrock1974}. Assuming regularity, the linear mapping between inputs $\nu$ and outputs $\psi$ that follows from \eqref{eq:FOM_laplace} is given by the system's transfer function  
	\begin{equation}
		\pHtf(s) := \pC(s\pE-\pA)^{-1}\pB+\pD.
	\end{equation}
	
	Since all transformations of the state-space equations in \eqref{eq:ss_def} can be expressed by operations on $\Ros$, the system matrix proved to be particularly useful for studying the properties of these transformations \cite{Rosenbrock1970}. In this work, we shall focus on transformations which leave the system's transfer function $\pHtf$ unchanged and are summarized by the notion of \emph{strict system equivalence}. 

	\begin{lemma} \cite{Rosenbrock1970} \label{lem:sse}
		Let $X(s)\in\R[s]^{m\times n}$, $Y(s)\in\R[s]^{n\times m}$ and define unimodular matrices $L(s),M(s)\in\R[s]^{n\times n}$, i.e. their determinants are nonzero constants. Suppose that two system matrices $\Ros$ and $\tRos$ are related by the transformation
		\begin{equation}
			\tRos(s) =  \Tone(s)\Ros(s)\Ttwo(s)=\mat{L(s)&0\\X(s)&I_\dimu} \Ros(s) \mat{M(s) & Y(s) \\ 0 & I_\dimu}.
		\end{equation}
		Then we shall say that $\Ros$ and $\tRos$ are related by \emph{strict system equivalence (s.s.e.)}. The two system matrices give rise to the same transfer function, i.e. ${\pHtf(s) = \pHtft(s)}$ for all ${s\in\C}$.		
	\end{lemma}
	\begin{proof}
		For a proof, we refer the reader to \cite[Section 3.1]{Rosenbrock1970}.
	\end{proof}	
	The benefits of representing a dynamical system with its system matrix are by no means restricted to system transformations; it may also be useful in the context of model reduction.	
	\subsection{Interpolatory model reduction}
	The goal of MOR is to find a reduced-order model
	\begin{equation}
		\label{eq:LTI_ROM}
			\begin{split}
				\rE \rdx(t) &= \rA \rx(t) + \rB u(t), \quad \rx(0) = 0,\\
				\ry(t) &= \rC \rx(t) + \rD u(t),
			\end{split}	
	\end{equation}
	with reduced state vector $\rx(t) \in\R^{r}$ and an associated transfer function $\pHtfr$ such that $r\ll\dimx$ and $\py\approx\ry$ for certain $u$.
	In projection-based MOR, these models are created by means of Petrov-Galerkin projections. Here, we define two matrices $U,\,V\in\R^{\dimx\times \dimxr}$ with full column rank. The original state trajectory $x(t)$ is approximated on the column space of $V$, i.e. $x(t)\approx V\rx(t)$ for all $t \in [0, \infty)$. The reduced system matrix is obtained by the following operation on the original system matrix $\Ros$:
	\begin{equation}\label{eq:projective_MOR}
		\rRos(s) =  \mat{U^\T&0\\0&I_\dimu} \Ros \mat{V & 0 \\ 0 & I_\dimu} = \mat{sU^\T\pE V-U^\T\pA V&-U^\T\pB\\\pC V&\pD} = \mat{s\rE-\rA&-\rB\\\rC&\rD}.
	\end{equation}
	The matrices $U,V$ may be chosen for different purposes, for example, to enforce the preservation of certain system-theoretic or structural properties of the original model. In interpolatory MOR methods, they are used to enforce interpolation conditions between the original and reduced transfer function. In particular, the column space of $V$ may be chosen such that the reduced transfer function $\pHtfr$ tangentially interpolates the original transfer function $H$ at a set of interpolation points or \emph{shifts} ${\{\sigma_1,\ldots,\sigma_r\} \in\C}$: 
	\begin{equation} \label{eq:interp_cond}
		H(\sigma_i)b_i=H_r(\sigma_i)b_i,\qquad i=1,\dots, r,
	\end{equation}
	where $b_i\in\C^{\dimu}$ denotes the associated (right) tangential direction. For the sake of simplicity, we assume throughout this work that each shift is distinct while an extension to multiple shifts is straightforward, and we refer the interested reader to \cite{Antoulas2020} for details. In all cases, we require ${\sigma_i\pE-\pA}$ as well as ${\sigma_i\rE-\rA}$ to be nonsingular for all ${i=1,\dots,r}$, and that the sets of shifts and associated tangential directions are both closed under complex conjugation. Then, the interpolation conditions in \eqref{eq:interp_cond} may be enforced by choosing
	\begin{equation} \label{eq:Inp_Krylov_PHS}
		(\sigma_i\pE-\pA)^{-1}\pB b_i \in \range(V),\qquad i=1,\dots, r,
	\end{equation}	
	where $\range(V)$ denotes the associated column space of $V$. Similarly, $U$ may be used to enforce additional (left) tangential interpolation conditions. The approximation quality of a ROM is typically assessed by computing the $\hinf$ or $\htwo$ norm of the error function $\pHtf-\pHtfr$. These are defined in the Hardy spaces $\rhinf^{m \times m}$ ( $\mathcal{RH}_2^{m \times m}$) of all proper (strictly proper) real-rational $m \times m$ matrices without poles in the closed complex right half-plane (see, e.\,g., \cite{ZhoDG96} for  details).
	\begin{remark}\label{rem:dae_challenge}
		 Note that while the conditions in \eqref{eq:Inp_Krylov_PHS} hold regardless of whether $\pE$ has full rank or not, as long as $\sigma_i\pE-\pA$ and $\sigma_i\rE-\rA$ are nonsingular for all $i=1,\dots,r$, the interpolation of DAE systems poses additional challenges. Unlike in the ODE case, where $\pHtf-\pHtfr$ is guaranteed to be strictly proper since ${\lim_{s\to\infty}\pHtf(s)-\pHtfr(s)=\pD-\rD = 0}$, this is not necessarily the case for DAE systems. In fact, as we will also see in Section \ref{sec:our_approach}, the algebraic constraints may even lead to \emph{improper} transfer functions with $\lim_{s\to\infty}H(s)=\infty$. It is, therefore, crucial to analyze the impact of the algebraic constraints on $\pHtf$ because otherwise, $\pHtf-\pHtfr$ may grow unboundedly large for $s\to\infty$, and the error norms are no longer defined. For an overview of how to deal with this challenge for general LTI systems, the reader is referred to \cite[Chapter 9]{Antoulas2020}.
	\end{remark} 
	
	\subsection{Port-Hamiltonian descriptor systems}
	The system class of linear port-Hamiltonian descriptor systems with quadratic Hamiltonian was introduced in \cite{BeaMXZ18}. In this work, we focus on the subclass of constant-coefficient pH-DAEs, which can be characterized by a special structure of the associated system matrix.	
	\begin{definition}\label{def:ph}
		Consider a regular LTI system of the form
		\begin{equation}\label{eq:FOM}
		\pHsys: 
		\begin{cases}
				\begin{aligned}
				\wt{E} \dot{\wt{x}}(t) &= (\wt{J}-\wt{R})\wt{x}(t) +(\wt{G}-\wt{P})u(t),\quad \px(0) = 0,\\
				y(t) &= {(\wt{G}+\wt{P})}^\T \wt{x}(t) +(\wt{S}+\wt{N})u(t),
			\end{aligned}
		\end{cases}
		\end{equation}
		where $\pE,\, \pJ,\, \pR \in \R^{\dimx \times \dimx}$, $\pG,\, \pP \in \R^{\dimx\times \dimu}$, $\pS,\, \pN \in \R^{\dimu\times \dimu}$. We call the system a pH-DAE system if its system matrix may be decomposed into the following sum of symmetric and skew-symmetric parts
		\begin{equation} \label{eq:ros_ph}
			\Ros(s) = s\underbrace{\mat{E&0\\0&0}}_{=:\,\mathcal{E}} + \underbrace{\mat{-\pJ & -\pG \\ \pG^\T & \pN}}_{=:\,\Gamma} + \underbrace{\mat{\pR & \pP \\ \pP^\T & \pS}}_{=:\,W},
		\end{equation}
		such that
		\begin{itemize}
			\item[(i)] the \emph{structure matrix} $\Gamma$ is skew-symmetric, i.e. $\Gamma = -\Gamma^\T$,
			\item[(ii)] the \emph{dissipation matrix} $W$ is positive semi-definite, i.e. $W = W^\T \geq 0$,			
			\item[(iii)] the \emph{extended descriptor matrix} $\mathcal{E}$ is positive semi-definite, i.e. $\mathcal{E} = \mathcal{E}^\T \geq 0$.
		\end{itemize}
	The system has an associated quadratic Hamiltonian $\mathcal{H}(x):=\frac{1}{2}\px^\T\pE\px$ and transfer function 
	\begin{equation*}
		H(s):=(G+P)^\T(sE-(J-R))^{-1}(G-P)+S+N.
	\end{equation*}
	\end{definition}
	
	Note that the definition proposed in \cite{BeaMXZ18} appears to be more general since it also allows the representation of systems governed by a quadratic Hamiltonian of the form ${\mathcal{H}(x) = \frac{1}{2}x^\T Q^\T E x}$ with $Q\in\R^{\dimx\times \dimx}$ and $\pQ^\T \pE = \pE^\T\pQ \geq 0$. However, our definition does not impose any additional restrictions since it has been shown, e.g., in~\cite{MehU22} that every pH-DAE, as defined in \cite{BeaMXZ18}, may be reformulated to have the form in \eqref{eq:FOM}.
	
	Moreover, it was shown in \cite{Achleitner2021,BeaGM21} that every pH-DAE may be transformed into \emph{staircase form}, which reveals, for example, the system's differentiation index. Before we proceed, let us highlight that the computation of this condensed form may require several subsequent full rank decompositions, which are sensitive to perturbations. Fortunately, due to the structural properties of the system matrix, the number of required decompositions is limited to three in contrast to general (unstructured) DAE systems \cite{Achleitner2021, MehU22}. Moreover, if this condensed form is directly considered during modelling, fewer steps are required, as discussed in \cite{Guducu2021, MehU22}. In some practical cases, for example, in the modelling of electric circuits, the staircase form even naturally arises or can be enforced by simple structure-preserving permutations of the system equations, as illustrated in Section \ref{sec:num_examples}.

	\begin{lemma}\label{lem:staircase} \textup{\cite{Achleitner2021,BeaGM21}}
		A regular pH-DAE system is in \emph{staircase form} if it has a partitioned state vector $\px(t) = {\left[\px_1(t)^\T,\px_2(t)^\T,\px_3(t)^\T,\px_4(t)^\T\right]}^\T$, where ${\px_j(t)\in\R^{\dimx_j}, \dimx_j\in \mathbb{N}_0}$ for all ${j=1,\ldots,4}$ such that
		\begin{alignat}{2}
			\pE &:= \mat{\pE_{11} & 0 & 0 & 0 \\ 0 & \pE_{22} & 0 & 0 \\ 0 & 0 & 0 & 0 \\ 0 & 0 & 0 & 0}, \quad	\pJ &&:= \mat{\pJ_{11} & \pJ_{12} & \pJ_{13} & \pJ_{14} \\  \pJ_{21} & \pJ_{22} & \pJ_{23} & 0 \\ \pJ_{31} & \pJ_{32} & \pJ_{33} & 0 \\ \pJ_{41} & 0 & 0 & 0}, \\ \pG &:= \mat{\pG_1 \\\pG_2 \\ \pG_3 \\ \pG_4}, \pP := \mat{\pP_1 \\\pP_2 \\ \pP_3 \\ 0}, \quad \pR &&:= \mat{\pR_{11} & \pR_{12} & \pR_{13} & 0 \\  \pR_{21} & \pR_{22} &\pR_{23} & 0 \\ \pR_{31} & \pR_{32} & \pR_{33} & 0 \\ 0 & 0 & 0 & 0}, 
		\end{alignat}
		where $\pE_{11},\,\pE_{22}$ are positive definite, and the matrices $\pJ_{41}$ and $\pJ_{33}-\pR_{33}$ are invertible (if the blocks are nonempty). The differentiation index $\ind$ of the uncontrolled system satisfies
		\begin{equation*}
			\ind = \begin{cases}
				0 & \textit{if and only if } \dimx_1 = \dimx_4 = 0 \textit{ and } \dimx_3 = 0, \\ 1 & \textit{if and only if } \dimx_1=\dimx_4=0 \textit{ and } \dimx_3 > 0, \\ 2 & \textit{if and only if } \dimx_1 = \dimx_4 > 0.
			\end{cases} 
		\end{equation*}
	\end{lemma}	
	As initially stated, it is beneficial to preserve the structural properties of the original pH-DAE model during MOR. Structure-preserving MOR methods, therefore, search for a reduced pH-DAE
	\begin{equation}\label{eq:ROM}
		\begin{aligned}
			\rE \rdx(t) = (\rJ-\rR)\rx(t) +(\rG-\rP)u(t), \\
			\ry(t) = {(\rG+\rP)}^\T \rx(t) +(\rS+\rN)u(t),
		\end{aligned}
	\end{equation}
	with $\rx(t)\in\R^\dimxr$, $r\ll n$ that fulfils the pH structural constraints, i.e. the associated system matrix may be decomposed as in Definition \ref{def:ph} such that
	\begin{equation} \label{eq:RB_props}
		\rRos(s) = s\red{\mathcal{E}}+\red{\Gamma}+\red{W},
	\end{equation}
	with symmetric positive semi-definite $\red{\mathcal{E}},\red{W}$ and skew-symmetric $\red{\Gamma}$. Note that the system matrix has also recently been used to derive a symplectic MOR method for LTI pH-ODEs without feedthrough in \cite{Zwart22}.
	
	\section{Our approach} \label{sec:our_approach}
	In the following, we demonstrate how the concepts of the presented staircase form and the system matrix may be unified to derive a framework for tangential interpolation of pH-DAEs with an arbitrary differentiation index. For this, we proceed in three steps. First, we apply a transformation under s.s.e. to the original system matrix $\Ros$ that enables us to decompose the original transfer function into proper parts and improper parts that may originate from algebraic constraints. Since all improper parts have to be preserved in the ROM exactly to keep the error $\pHtf-\pHtfr$ bounded (see Remark~\ref{rem:dae_challenge}), we propose a new method to efficiently reduce \emph{only} the proper part in the second step. Third, we show how to reattach the original improper part to the reduced proper transfer function to construct a minimal pH-DAE representation of the ROM in staircase form.
	
	\subsection{System matrix decomposition}\label{subsec:trafo}
	Let $\Ros$ denote the system matrix of a full-order pH-DAE system with transfer function $\pHtf$ in staircase form as in Lemma \ref{lem:staircase}. For the sake of notational simplicity, we use ${A=J-R}$, ${B=G-P}$, ${C=(G+P)^\T}$ and ${D=S+N}$ which are partitioned as in Lemma \ref{lem:staircase}: $A_{11}\in\R^{n_1\times n_1}$, for example, denotes the upper left block of ${J-R}$.  We define the transformation matrices ${\Tone,\Ttwo\in\R^{(\dimx+\dimu)\times(\dimx+\dimu)}}$ with
	\begin{align*}
		\Tone &:= \mat{L&0\\X&I_m} = \text{\footnotesize $\mat{I_{\dimx_1} & 0 & -\pA_{13}\pA_{33}^{-1} & 0 & 0 \\
				0 & I_{\dimx_2} & -\pA_{23}\pA_{33}^{-1} & -(\pA_{21}-\pA_{23}\pA_{33}^{-1}\pA_{31})\pA_{41}^{-1} & 0 \\ 
				0 & 0 & \pA_{33}^{-1} & 0 & 0 \\
				0 & 0 & 0 & A_{41}^{-1} & 0 \\
				\pC_4\pA_{14}^{-1} & 0 & (\pC_3-\pC_4\pA_{14}^{-1}\pA_{13})\pA_{33}^{-1} & 0 & I_{\dimu}}$}, \\
		\Ttwo &:= \mat{M&Y\\0&I_m} =  \text{\footnotesize $\mat{I_{\dimx_1} & 0 & 0 & 0 & -\pA_{41}^{-1}\pB_4 \\
				0 & I_{\dimx_2} & 0 & 0 & 0 \\ 
				-\pA_{33}^{-1}\pA_{31} & -\pA_{33}^{-1}\pA_{32} & I_{n_3} & 0 & -\pA_{33}^{-1}(\pB_3-\pA_{31}\pA_{41}^{-1}\pB_4) \\
				0 & \pA_{14}^{-1}(-\pA_{12}+\pA_{13}\pA_{33}^{-1}\pA_{32}) & 0 & \pA_{14}^{-1} & 0 \\
				0 & 0 & 0 & 0 & I_{\dimu}}$}.
	\end{align*}	
	It is apparent that $\Tone$ and $\Ttwo$ satisfy the conditions in Lemma \ref{lem:sse} since the determinants of $\pA_{33}, \pA_{14}$ and $\pA_{41}$ are constant and nonzero. Note that our definitions of $L$ and $M$ share similarities with the state-space transformation presented in \cite[Lemma 6]{Achleitner2021} for autonomous semi-dissipative Hamiltonian DAEs. We obtain a transformed s.s.e. system matrix 
	\begin{equation} \label{eq:sse_trafo}
		\tRos(s) = \Tone\Ros(s)\Ttwo = \mat{s\tE_{11}-\tA_{11} & 0 & 0 & -I_{\dimx_1} & -\tB_1 \\
			0 & \boldsymbol{s\prE-\prA} & 0 & 0 & -\boldsymbol{\prB} \\ 
			0 & 0 & -I_{\dimx_3} & 0 & 0 \\
			-I_{\dimx_1} & 0 & 0 & 0 & 0 \\
			\tC_1 & \boldsymbol{\prC} & 0 & 0 & \boldsymbol{\prD+s\Dinf}},
	\end{equation}
	with nonsingular $\tE_{11},\,\prE$. The second and fifth block column and row entries, respectively, are highlighted since these are the only parts that contribute to the transformed transfer function $\pHtft$:
	\begin{align} \label{eq:pHtf_decomp}
		\pHtft &= \tC(s\tE-\tA)^{-1}\tB + \prD + s\Dinf \\
		&= \text{ \footnotesize $\mat{\tC_1 & \prC & 0 & 0}\mat{0 & 0 & 0 & -I_{n_1} \\ 0 & (s\prE-\prA)^{-1} & 0 & 0 \\ 0 & 0 & -I_{n_3} & 0 \\ -I_{n_1} & 0 & 0 & -(s\tE_{11}-\tA_{11})}\mat{\tB_1 \\ \prB \\ 0 \\ 0} + \prD + s\Dinf$} \\
		&= \prC(s\prE-\prA)^{-1} \prB + \prD + s\Dinf
	\end{align}
	
	Consequently, the transfer function of every pH-DAE may be represented as the sum of the transfer function of a \emph{proper} ODE system with dimension $n_2$ and an additional linear, improper term. Moreover, as shown in~\cite[Theorem 1]{MosSMV2022b}, the proper subsystem again satisfies the pH structural constraints in Definition \ref{def:ph}. It, therefore, admits a pH-ODE representation, which can be found by using its system matrix.
	\begin{lemma}\label{lem:prop_decomp}
		A pH-ODE representation of the proper subsystem with transfer function $\pHtfp = \prC(s\prE-\prA)^{-1} \prB +\prD$ can be computed by simply decomposing its system matrix
		\begin{equation*}
			\prRos(s) := \mat{s\prE-\prA&-\prB\\\prC&\prD}, 
		\end{equation*}		
		into symmetric and skew-symmetric parts, respectively.
	\end{lemma}
	\begin{proof}
		Decomposing $\prRos$ yields 
		\begin{equation*}
			\prRos(s) = s\mat{\prE&0\\0&0} + \prGamma + \prW,
		\end{equation*}	
		with
		\begin{align*}
			\prGamma  &= \mat{-\prJ & -\prG \\ \prG^\T & \prN} = \frac{1}{2} \mat{-\prA+\prA^\T & -\prB-\prC^\T\\\prC+\prB^\T&\prD-\prD^\T},\\
			\prW &= \mat{\prR & \prP \\ \prP^\T & \prS} = \frac{1}{2} \mat{-\prA-\prA^\T&-\prB+\prC^\T\\\prC-\prB^\T&\prD+\prD^\T}. \\
		\end{align*}
		The fact that the system is an ODE system follows directly from ${\prE = \pE_{22}>0}$. This also proves condition (i) in Definition~\ref{def:ph}. In~\cite[Theorem 1]{MosSMV2022b}, it was shown that the sum ${\prGamma+\prW}$ may be obtained from a series of transformations of the original ${\Gamma+W}$. These include permutations, Schur complement constructions and congruence-like transformations, which all preserve the positive semi-definiteness of the symmetric part. Therefore, $\prGamma, \prW$ fulfil the pH structural constraints (ii) and (iii) in Definition~\ref{def:ph}, which completes the proof.
	\end{proof}

	We highlight that the simplicity of this result is a direct consequence of the staircase form and the pH structural constraints. For general DAE systems, this system decomposition approach generally requires the computation of spectral projectors onto the left and right deflating subspaces of the pencil $\lambda\pE -\pA$ corresponding to the finite eigenvalues, which are numerically challenging to compute in the large-scale setting \cite{Gugercin2013}. In some applications, such as fluid flow problems or electric circuit simulation where the matrices $\pE$ and $\pA$ have a special block structure, the computation of spectral projectors can be done more efficiently or even circumvented, see, e.g., \cite{MehS05,Gugercin2013}. However, the proposed interpolatory MOR approaches for general (unstructured) DAE systems \cite{Gugercin2013,Antoulas2020} vary for different differentiation indices, and their adaptations to pH-DAE systems proposed in \cite{BeaGM21, HauMM19} do not always guarantee that the ROM is again in pH form. In the following, we show how the results in this section enable the construction of a general MOR approach that works irrespective of the original system's differentiation index and guarantees to produce minimal ROMs in pH form.
	
	\subsection{Tangential interpolation}\label{subsec:mor}
	
	As discussed in Remark \ref{rem:dae_challenge}, the ROM has to match the improper part $s\Dinf$ exactly because otherwise,  $\pHtf-\pHtfr \notin \rhinf^{m \times m}$. We will therefore set this part aside for now and focus on the reduction of the proper part. Since the proper subsystem has only dimension $n_2$, a natural approach would be to reduce the proper system matrix $\prRos$ directly. According to~\eqref{eq:Inp_Krylov_PHS}, this approach requires the computation of solutions $v_i\in\C^{n_2}$ of
	\begin{equation}\label{eq:proper_Krylov}
		(\sigma_i\prE-\prA)v_i = \prB b_i,
	\end{equation} 
	for all $i=1,\,\ldots\,r$. From \eqref{eq:sse_trafo}, we derive
	\begin{align*}
		\prE = \pE_{22}, \qquad \prA = \pA_{22}-\pA_{23}\pA_{33}^{-1}\pA_{32}.
	\end{align*}
	For index-1 pH-DAEs ($n_3>0$) with nonzero $\pA_{23}\pA_{33}^{-1}\pA_{32}$, the matrix $\prA$ might be dense, and therefore, the solutions in \eqref{eq:proper_Krylov} may be more expensive than for the original system with matrices $E$, $A$ and $B$. This is illustrated by the following example.
	\begin{example}\label{ex:sparsity}
	Consider a pH-DAE with differentiation index $\nu = 1$ in staircase form. The system has dimension $\dimx=10^4$ with $\dimx_2=\dimx_3=\frac{n}{2}$ and one input-output pair. As shown by the sparsity patterns in Figure \ref{fig:sparsity} (a)-(b), the matrices $\pE$ and $\pA$ have a very simple structure with only few nonzero entries (depicted in blue). However, the matrix $\prA$ of the proper subsystem is dense, as shown in Figure \ref{fig:sparsity} (d). 
	\begin{figure}[tbh]%
		\centering
			\subfloat[]{\includegraphics[width=0.45\textwidth]{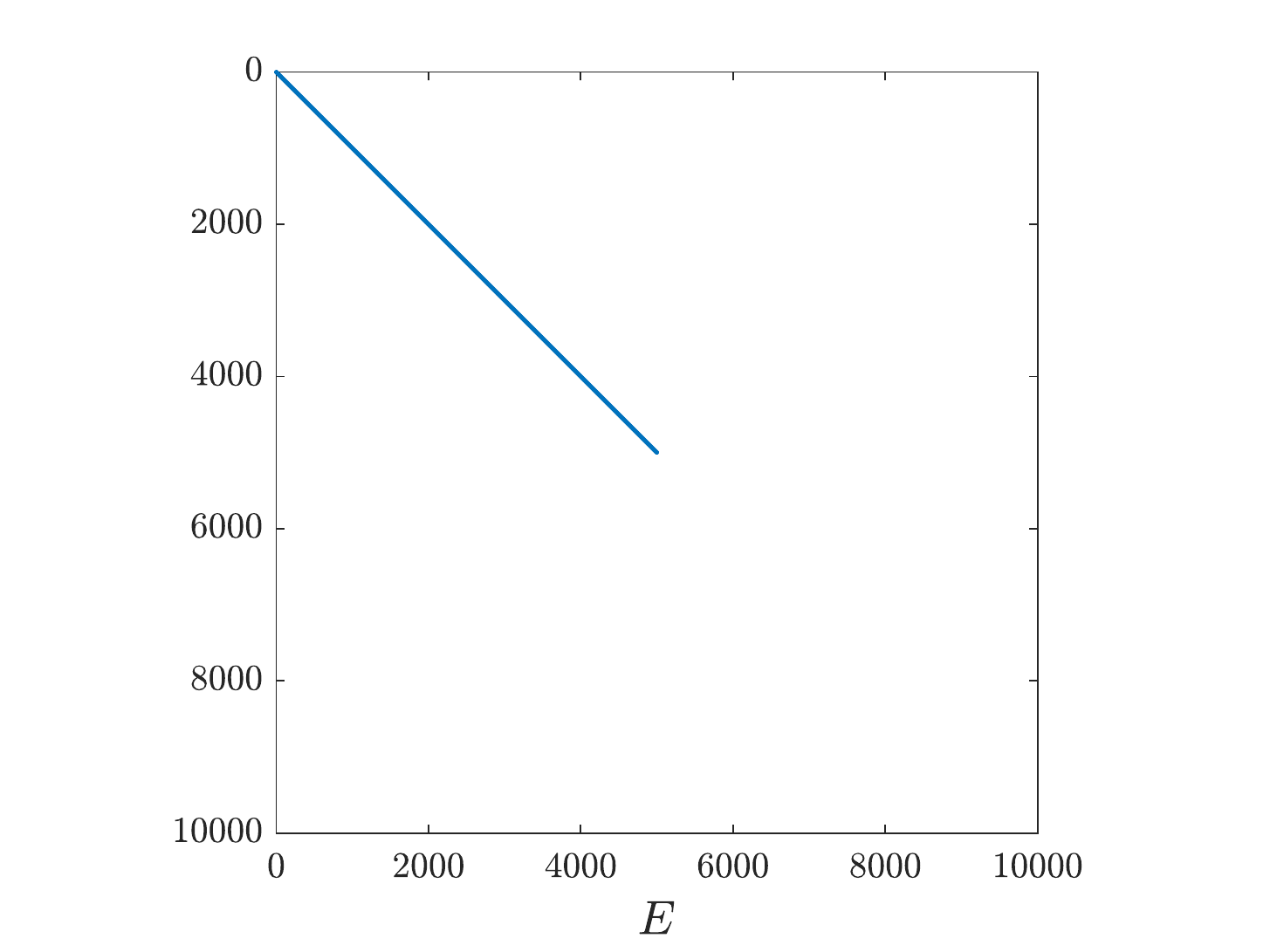}}
			\subfloat[]{\includegraphics[width=0.45\textwidth]{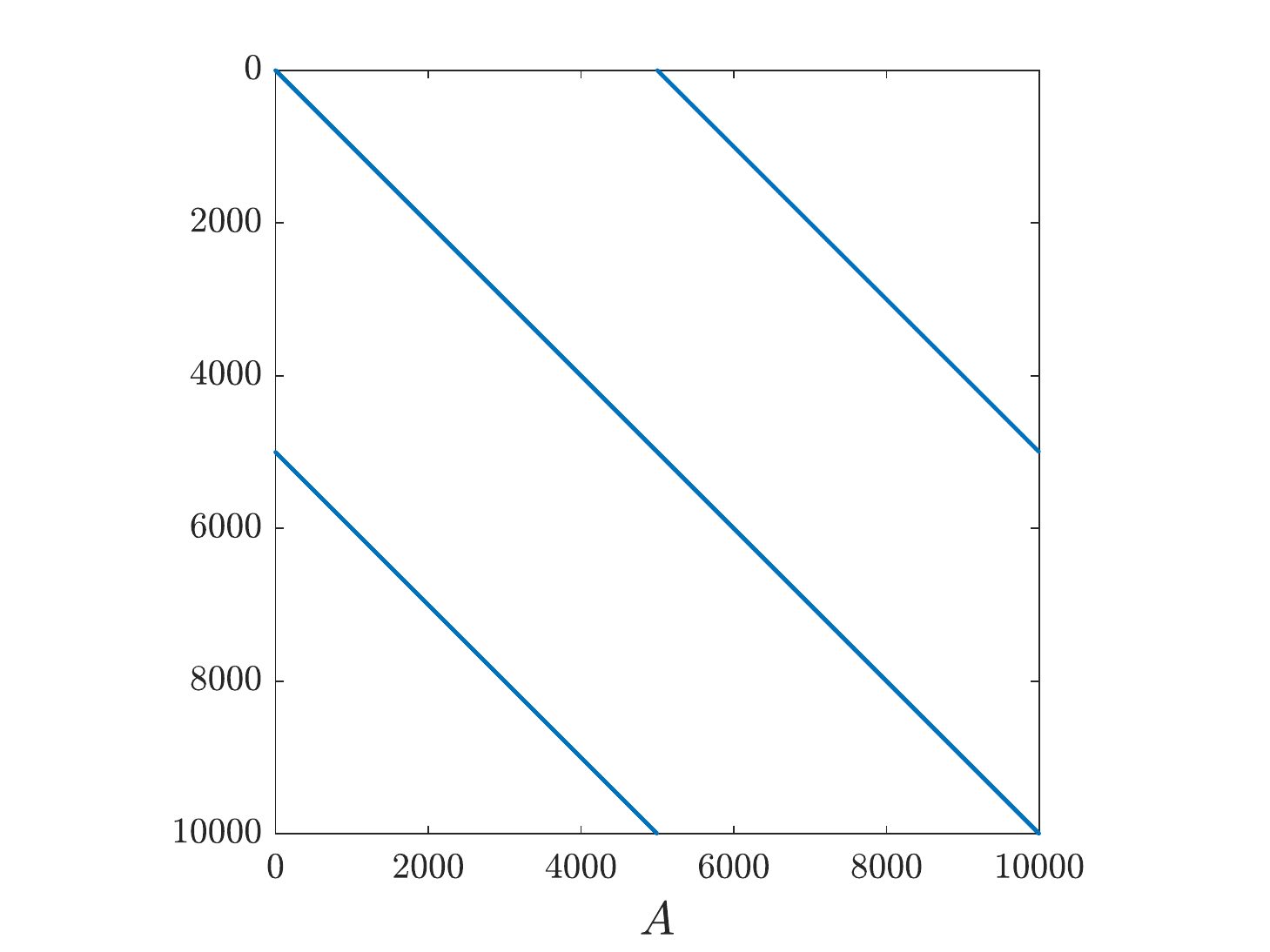}} \\
			\subfloat[]{\includegraphics[width=0.225\textwidth]{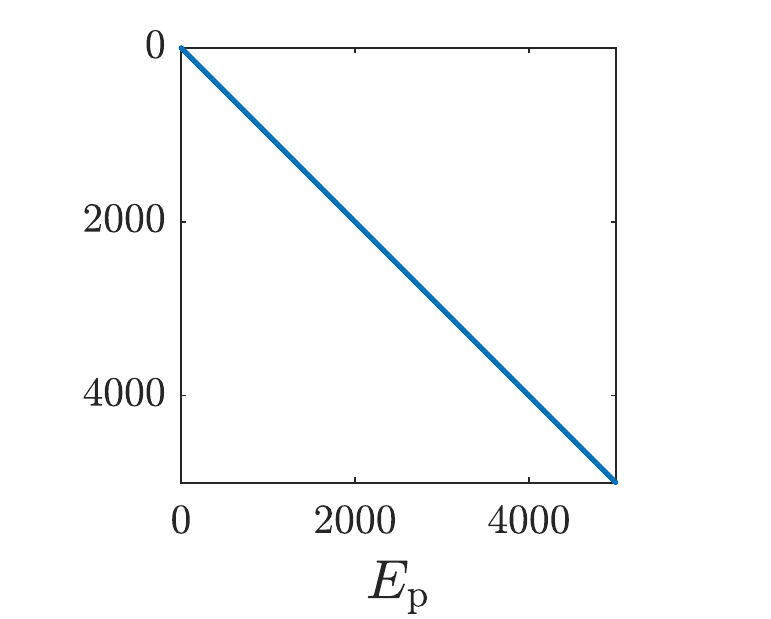}} \hspace{90pt}
			\subfloat[]{\includegraphics[width=0.225\textwidth]{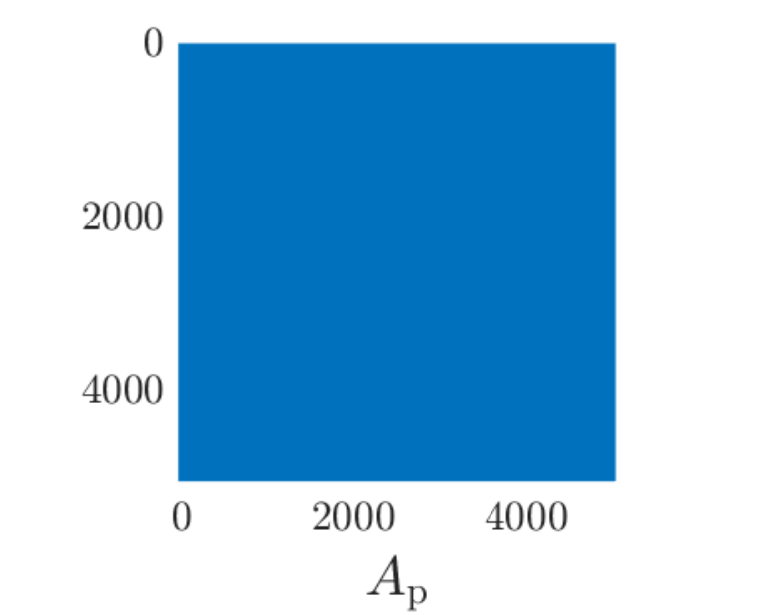}}
		\caption{Sparsity patterns of exemplary full-order matrices $\pE,\pA\in\R^{n\times n}$ and of the matrices $\prE,\prA\in\R^{n_2\times n_2}$ of its associated proper subsystem.}
		\label{fig:sparsity}
	\end{figure}
	To demonstrate the effects on interpolatory MOR, we solved \eqref{eq:proper_Krylov} for the proper subsystem and for the original system with matrices $E$, $A$ and $B$ using $100$ random complex shifts $\sigma_i$ with MATLAB's \texttt{mldivide} command. All computations were conducted using MATLAB R2021b (version 9.11.0.1873467) on an Intel\textsuperscript{\textregistered} Core\texttrademark  i7-8700 CPU (3.20 GHz, 6-Core) with 32 GB RAM. The computations of \eqref{eq:proper_Krylov} for the proper subsystem took $2.51$ seconds on average versus only $0.013$ seconds for the original model. Consequently, even though the proper subsystem is significantly smaller in size, the computation of the reduction matrix $V$ takes more than 150 times longer. 

	\end{example}

	Therefore, we propose another approach that works with the original (sparse) system matrix, irrespective of the system's differentiation index. 	
	\begin{theorem}
		Given a large-scale pH-DAE in staircase form with system matrix $\Ros$, as well as interpolation points $\{\sigma_1,\dots,\sigma_r\}$, and corresponding right tangential directions $\{b_1,\dots,b_r\}$, let $V\in\R^{n\times r}$ define a reduction matrix such that \eqref{eq:Inp_Krylov_PHS} holds with a decomposition $V^\T = [V_1^\T,V_2^\T,V_3^\T,V_4^\T]^\T$ with $V_j \in \R^{n_j\times r}$ for all $j=1,\ldots,4$. We define the reduction matrices	
		\begin{equation}\label{eq:red_matrices}
			\We = 	\mat{0&A_{14}^{-\T}C_4^\T\\ V_2& 0\\-A_{33}^{-\T}A_{23}^\T V_2 & 	A_{33}^{-\T}(C_3^\T-A_{13}^\T A_{14}^{-\T} C_4^\T)\\0&0\\0&I_m},\quad \Ve = \mat{0& A_{14}^{-\T}B_4\\V_2& 0 	\\ 0 & 0 \\0&0\\0&I_m}, 
		\end{equation}
		Then, the model associated with the reduced system matrix
		\begin{equation*}
			\rRos := \We^\T \Ros \Ve,
		\end{equation*}
		satisfies the tangential interpolation conditions \eqref{eq:interp_cond}, and its proper subsystem fulfils the pH structural constraints. The reduced system matrix admits a decomposition
		\begin{equation*}
			\rRos = s\mat{\rprE & 0 \\ 0 & \Dinf} + \mat{-\rprJ & -\rprG \\ \rprG^\T & \rprN} + \mat{\rprR & \rprP \\ \rprP^\T & \rprS}.
		\end{equation*}
	\end{theorem}
	\begin{proof}
	Using the decomposition of $V$, the original system matrix $\Ros$ may initially be reduced in the following way:
	\begin{equation} \label{eq:Galerkin}
		\Ros_{2}(s) =  \mat{U^\T&0\\0&I_{\dimu}} \Ros \mat{U&0\\0&I_{\dimu}}, \qquad U:= \text{\footnotesize$\mat{I_{n_1}&0&0&0\\0&V_2&0&0\\0&0&I_{n_3}&0\\0&0&0&I_{n_4}}$}.
	\end{equation}	
	The reduced model associated with $\Ros_{2}$ satisfies the tangential interpolation conditions in \eqref{eq:interp_cond} since $\text{range}(V)\subseteq \text{range}(U)$. Note that this approach, which was also proposed in \cite{BeaGM21} for index-1 pH-DAEs, produces ROMs that are still comparatively large: the ROM has dimension $n_1+r+n_3+n_4$. However, since the ROM is still a pH-DAE in staircase form, we can obtain a minimal representation by applying a transformation under s.s.e., as in Section \ref{subsec:trafo}, and extracting the proper system matrix without changing its transfer function. A combination of the reduction in \eqref{eq:Galerkin}, the transformation and extraction of proper parts yields the reduction matrices $\We, \Ve$. Moreover, simple algebraic manipulations show that
	\begin{equation*}
		\rRos = \We^\T \Ros \Ve = \mat{V_2^\T&0\\0&I_{\dimu}}\mat{s\prE-\prA&-\prB\\\prC&\prD+s\Dinf}\mat{V_2&0\\0&I_{\dimu}},
	\end{equation*}
	and consequently, the proper part of $\rRos$ fulfils the pH structural constraints and may be decomposed to obtain a pH representation, which completes the proof.	
	\end{proof}

	To improve the numerical stability of Krylov subspace methods, the matrix $V$ is usually orthogonalized such that $V^\T V=I_r$. This orthogonalization does not change the moment matching conditions in \eqref{eq:Inp_Krylov_PHS} since these only depend on the subspace that is spanned by the column vectors of $V$, not the basis itself. However, note that in our case, even if $V$ is orthogonal, this is generally not the case for its submatrix $V_2$. To improve the numerical stability, we employ the cosine-sine decomposition, as discussed in \cite{EggKLMM18}. For this, we split $V$ into two parts: $V_2$ and the remaining submatrices ${V_{\rm rem}^\T = [V_1^\T,V_3^\T,V_4^\T]^\T\in\R^{n_{\rm rem}\times r}}$. We then compute the decomposition
	\begin{equation} \label{eq:sin_cos}
		\mat{V_2\\V_{\rm rem}} = \mat{\bar{V}_2&0\\0&\bar{V}_{\rm rem}}\mat{C_s\\S_s}X_s^\T,
	\end{equation} 
	with orthogonal $\bar{V_2}\in\R^{n_2\times r}$, $\bar{V}_{\rm rem}\in\R^{n_{\rm rem}\times r}$, and $X_s\in\R^{r\times r}$, as well as diagonal ${C_s,S_s\in\R^{r\times r}}$ such that ${C_s^\T C_s + S_s^\T S_s = I_r}$. Replacing $V_2$ in \eqref{eq:red_matrices} with $\bar{V}_2$ yields the final reduced system matrix $\rRos$. 
	
	\subsection{Minimal pH-DAE representation}\label{subsec:compose_ROM}
	
	To find a pH-DAE representation for $\rRos$, we have to incorporate the improper part $s\Dinf$ that has been separated back into the model. Two different methods have been proposed in \cite{MosSMV2022b} and \cite{Cherifi2022} for this purpose. Since the method in \cite{Cherifi2022} only leads a minimal ROM representation if $\Dinf$ has full rank, we proceed similarly as in \cite{MosSMV2022b}. We have that  
	\begin{equation}\label{eq:dinf_prop}
		\Dinf = -\pC_4\pA_{14}^{-1}\pE_{11}\pA_{41}^{-1}\pB_4 = \pG_4^\T\pA_{41}^{-\T}\pE_{11}\pA_{41}^{-1}\pG_4 = \Dinf^\T \geq 0.
	\end{equation} 
	Consequently, there exists a rank-revealing factorization ${\Dinf = L_{\infty}L_{\infty}^\T}$ with ${L_{\infty}\in \R^{\dimu\times q}}$. A minimal ROM representation $\pHsysr$ in staircase form can be found with
	\begin{alignat*}{3}
		\rE&= \mat{I_q  & 0 & 0 \\ 0 & \rprE  & 0 \\ 0 & 0 & 0},\quad &&\rJ=\mat{0 & 0 & I_q \\ 0 & \rprJ & 0  \\ -I_{q} & 0  & 0}, \quad &&\rR=\mat{0 & 0 & 0 \\ 0 & \rprR & 0 \\ 0 & 0 & 0}, \\ \rG &= \mat{0 \\ \rprG \\ L_\infty^\T}, &&\rP= \mat{0 \\ \rprP \\ 0},&&\rS=\rprS,\quad \rN=\rprN.
	\end{alignat*}

	\begin{algorithm}[htpb]
		\LinesNumbered
		\SetAlgoLined
		\DontPrintSemicolon
		\SetKwInOut{Input}{Input}\SetKwInOut{Output}{Output}
		\Input{Large-scale pH-DAE $\pHsys$ in staircase form with system matrix $\Ros$; set of interpolation points $\{\sigma_1,\dots,\sigma_r\}$ and corresponding right tangential directions $\{b_1,\dots,b_r\}$ (both closed under complex conjugation).
		}
		\Output{Reduced pH-DAE $\pHsysr$ with system matrix $\rRos$. \;}
		Compute $V = \left[V_1^\T,V_2^\T,V_3^\T,V_4^\T\right]^\T\in\R^{n\times r}$ such that \eqref{eq:Inp_Krylov_PHS} holds.\;
		Orthogonalize $V_2$ via the cosine-sine decomposition in \eqref{eq:sin_cos} such that  $\rm{range} (V_2) = \rm{range} (\bar{V}_2)$ with $\bar{V}_2^\T\bar{V}_2=I_r$. \;
		Compute the reduction matrices $\We, \Ve$ as in Section \ref{subsec:mor}:
		\begin{equation*}
			\We = 	\mat{0&A_{14}^{-\T}C_4^\T\\ \bar{V}_2& 0\\-A_{33}^{-\T}A_{23}^\T \bar{V}_2 & 	A_{33}^{-\T}(C_3^\T-A_{13}^\T A_{14}^{-\T} C_4^\T)\\0&0\\0&I_m},\quad \Ve = \mat{0& A_{14}^{-\T}B_4 \\ \bar{V}_2& 0 	\\ 0 & 0 \\0&0\\0&I_m}. 
		\end{equation*} \;
		\vspace*{-.4cm}
		Compute and decompose the reduced system matrix $$\Ros_r = \We^\T \Ros\Ve = s\mat{\rprE & 0 \\ 0 & \Dinf} + \mat{-\rprJ & -\rprG \\ \rprG^\T & \rprN} + \mat{\rprR & \rprP \\ \rprP^\T & \rprS}.$$ \;
		\vspace*{-.4cm}
		Compute a rank-revealing factorization $\Dinf = L_{\infty} L_{\infty}^\T$ with $L_{\infty}\in\R^{m\times q}$.\;
		\eIf {$q>0$} {
		Construct ROM $\pHsysr$ as in \eqref{eq:ROM} with
		\begin{alignat*}{3}
			\rE&= \mat{I_q  & 0 & 0 \\ 0 & \rprE  & 0 \\ 0 & 0 & 0},\quad &&\rJ=\mat{0 & 0 & I_q \\ 0 & \rprJ & 0  \\ -I_{q} & 0  & 0}, \quad &&\rR=\mat{0 & 0 & 0 \\ 0 & \rprR & 0 \\ 0 & 0 & 0}, \\ \rG &= \mat{0 \\ \rprG \\ L_\infty^\T}, &&\rP= \mat{0 \\ \rprP \\ 0},&&\rS=\rprS,\quad \rN=\rprN.
		\end{alignat*} 
		\vspace*{-.4cm}
		\;}
		{Construct ROM $\pHsysr$ as in \eqref{eq:ROM} directly from $\rRos$.}
		\caption{Tangential Interpolation of pH-DAEs}
		\label{alg:arnoldi}
	\end{algorithm}
	\begin{table}[htpb]
	\begin{center}
		\begin{tabular}{ |c||c|c|c|c|c|p{4cm}| } 
			\hline
			Category & $\nu$ & $n_2$ & $n_3$ & $n_4/n_1$ & $\Dinf$ & References\\
			\hline 
			\hline
			Index-0 & $0$ & $\neq 0$& $0$ & $0$ & $0$ & \cite[Theorem 7]{Gugercin2012} \\
			\hline 
			Index-1 & $1$ & $\neq 0$& $\neq 0$ & $0$ & $0$ & \cite[Theorems 1 and 2]{BeaGM21} \newline \cite[Theorem 5]{HauMM19} \\		
			\hline 
			Proper Index-2 & $2$ & $\neq 0$ & $0$ & $\neq 0$ & $0$ & \cite[Theorem 3]{BeaGM21} \\		
			\hline 
			Improper Index-2 & $2$ & $\neq 0$& $0$ & $\neq 0$ & $\neq 0$ & \cite[Theorem 4]{BeaGM21}  \\		
			\hline 
			Proper Index-1-2 & $2$ & $\neq 0$& $\neq 0$ & $\neq 0$ & $0$ & n/a \\		
			\hline 
			Improper Index-1-2 & $2$ & $\neq 0$ & $\neq 0$ & $\neq 0$ & $\neq 0$ & n/a  \\		
			\hline 
		\end{tabular}
		
		\caption{\label{tab:soa} Proposed tangential interpolation methods for pH-DAEs. The six different categories are derived from the original system's differentiation indices $\nu$ and whether its transfer function has improper parts.}
	\end{center}
\end{table}

	\subsection{Discussion}
	
	In the following, we would like to briefly discuss the differences between the proposed Rosenbrock framework and other interpolatory MOR approaches that have been proposed for pH-DAEs. On the one hand, the approaches get more complex if the system's differentiation index $\nu$ increases. On the other hand, as revealed in this section, if the system's transfer function has improper parts, these require special care since they have to be preserved exactly in the ROM. As shown in Lemma \ref{lem:staircase}, the differentiation index of LTI pH-DAEs is at most two, and improper parts may only occur if the differentiation index is two (see \eqref{eq:dinf_prop}). Therefore, in the context of MOR, we may identify six different system categories. An overview of existing methods for each category, to the best of the author's knowledge, is given in Table \ref{tab:soa}. 
	
	The tangential interpolation of pH-ODEs ($\nu=0$), as initially proposed in \cite{Gugercin2012}, is well understood and leads to minimal ROM representations in pH-ODE form. For these systems, our approach is equivalent since the reduction matrices simplify to ${\We=\Ve = \text{\footnotesize $\mat{\bar{V}_2&0\\0&I_\dimu}$}}$. For index-1 pH-DAEs, methods in \cite{BeaGM21,HauMM19} have been proposed, which rely on different semi-explicit representations of the FOM. In \cite[Theorem 1]{BeaGM21}, the feedthrough matrix of the ROM is modified to match the polynomial part of the FOM, which in the index-1 case, is constant. However, this method does not guarantee that the ROM fulfils the pH structural constraints. A remedy to this problem is to preserve the algebraic constraints of the FOM as proposed in \cite[Theorem 2]{BeaGM21} and \cite[Theorem 5]{HauMM19}. However, this does not generally yield minimal ROMs since redundant algebraic equations cannot be removed, as discussed in \cite[Remark 2]{BeaGM21}. In contrast, our method does not impose additional assumptions since it also only requires a semi-explicit FOM representation but guarantees the preservation of the pH form and yields minimal ROM representations. 
	
	For systems with $\nu = 2$ that do not have an index-1 part ($n_3=0$), we may distinguish between the proper and improper case. The proper index-2 case is comparable to the pH-ODE case if the system is in semi-explicit form and the method proposed in \cite{BeaGM21} yields minimal ROMs in pH form. Our method achieves the same goals, but may require one additional transformation to identify the full-rank matrix $\pJ_{41}$. The improper index-2 case is treated in \cite[Theorem 4]{BeaGM21}. However, this method does not guarantee that the ROM fulfils the pH structural constraints, which is generally the case for our method, again under the assumption that the system is in staircase form. To the best of the author's knowledge, methods for systems with index-1 \emph{and} index-2 parts ($n_3 \neq 0$, $n_4 \neq 0$) have not been proposed yet. These two categories are also covered by our framework, which applies to all pH-DAEs in staircase form. 
	
	\section{$\htwo$- and $\hinf$-inspired tangential interpolation} \label{sec:extensions}
	The question that remains is how to choose the parameters in Algorithm \ref{alg:arnoldi}, i.e. the set of interpolation points and tangential directions. For unstructured ODE and DAE systems, different $\htwo$- and $\hinf$-inspired strategies have been proposed (see, e.g., \cite{Gugercin2008,Druskin2011,Druskin2014,Flagg2013,Castagnotto2017}) which have also partly been adapted to the pH-ODE and pH-DAE cases in \cite{Gugercin2012,BeaGM21}. 
	\newpage
	In the following, we demonstrate how these ideas may be incorporated into the proposed Rosenbrock framework, and refer to the cited work in each section for implementation details.
	 
	\subsection{Interpolatory $\htwo$ approximation}
	
	In $\htwo$-inspired interpolation methods for general DAE systems, the transfer function of the FOM is typically decomposed into the sum ${H(s) = \pHtfsp(s)+\pHtf_{\rm pol}}$ with a \emph{strictly proper} part $\pHtfsp$ satisfying $\lim_{s\to\infty}\pHtfsp(s)=0$ and a polynomial part $\pHtf_{\rm pol}$ that potentially grows polynomially for $s\to\infty$. For the $\htwo$ error $\Vert\pHtf-\pHtfr\Vert_{\htwo}$ to be bounded, we require $\pHtf-\pHtfr \in \mathcal{RH}_2^{m \times m}$. Assuming a similar decomposition of $\pHtfr$ into ${\pHtfr(s) = \pHtf_{\rm sp,r}(s)+\pHtf_{\rm pol,r}(s)}$, this is only the case if $\pHtf_{\rm pol,r}=\pHtf_{\rm pol}$. The necessary conditions for locally $\htwo$-optimal ROMs may then be formulated as interpolation conditions.
	
	\begin{lemma}\label{lem:h2_opt_cond}{\cite[Theorem 9.2.1]{Antoulas2020}} Let the transfer function ${H(s) = \pHtfsp(s)+\pHtf_{\rm pol}}$ be decomposed into strictly proper and polynomial parts. Let the ROM transfer function $\pHtfr$ have an analogous decomposition ${\pHtfr(s) = \pHtf_{\rm sp,r}(s)+\pHtf_{\rm pol,r}(s)}$ with strictly proper part ${\pHtf_{\rm sp,r}(s)=C_{\rm sp,r}(sE_{\rm sp,r}-A_{\rm sp,r})^{-1} B_{\rm sp,r}}$ such that ${A_{\rm sp,r}\in\R^{r\times r}}$, nonsingular ${E_{\rm sp,r}\in\R^{r\times r}}$ and ${B_{\rm sp,r}, C_{\rm sp,r}^\T\in\R^{r\times m}}$. If $\pHtfr$ minimizes the $\htwo$ error $\Vert\pHtf-\pHtfr\Vert_{\htwo}$ over all ROMs with an $r$-th order strictly proper part, then $\pHtf_{\rm pol,r}=\pHtf_{\rm pol}$ and $\pHtf_{\rm sp,r}$ minimizes the $\htwo$ error $\Vert\pHtf_{\rm sp}-\pHtf_{\rm sp,r}\Vert_{\htwo}$. Let $\pHtf_{\rm sp,r}$ be represented by its pole-residue expansion ${\pHtf_{\rm sp,r}(s)=\sum\limits_{i=1}^{r} \frac{l_i r_i^\T}{s-\lambda_i}}$ where ${r_i\in \mathbb{C}^m, l_i\in\mathbb{C}^m}$ and with simple poles $\lambda_i \in \mathbb{C}$. Then, the tangential interpolation conditions
		\begin{subequations}
			\begin{align}
				\pHtf(-\lambda_i)r_i &= \pHtfr(-\lambda_i)r_i, \label{eq:OC1}\\
					l_i^\T\pHtf(-\lambda_i) &= l_i^\T\pHtfr(-\lambda_i), \label{eq:OC2}\\
				l_i^\T\pHtf^{\prime}(-\lambda_i)r_i &= l_i^\T\pHtfr^{\prime}(-\lambda_i)r_i, \label{eq:OC3}
			\end{align} \label{eq:OC}
		\end{subequations}			
	hold for all $i=1,\,\dots\,,r$.
	\end{lemma}	
	This connection between interpolatory and $\mathcal{H}_{2}$-optimal model reduction is the motivation behind the well-known \emph{iterative rational Krylov algorithm} (IRKA) \cite{Gugercin2008}, which utilizes a fixed-point iteration to enforce the necessary $\mathcal{H}_{2}$ optimality conditions for general DAE systems. Since for pH systems, the matrix $U$ in \eqref{eq:projective_MOR} is typically chosen to enforce structure preservation, fewer degrees of freedom are available for interpolation, and it is generally only possible to fulfil a subset of \eqref{eq:OC}, e.g. the conditions in \eqref{eq:OC1}. Enforcing these conditions in an iterative manner using \eqref{eq:Inp_Krylov_PHS} leads to the IRKA-PH algorithm proposed for pH-ODE systems in \cite{Gugercin2012}. 

	Embedding the IRKA-PH algorithm into the proposed pH-DAE framework is straightforward. In each IRKA-PH iteration, we first compute the reduced system matrix $\rRos$, using some initial interpolation data in the first iteration. We directly obtain the reduced strictly proper transfer function $$\pHtf_{\rm sp,r}(s)=(\rprG+\rprP)^\T(s\rprE-(\rprJ-\rprR))^{-1}(\rprG-\rprP).$$ Suppose that the pencil $\lambda\rprE-(\rprJ-\rprR)$ has simple eigenvalues $\lambda_i$ and let $t_i\in\C^r$ denote a left eigenvector associated with $\lambda_i$, i.e.
	\begin{equation} \label{eq:evp}
		t_i^\T(\lambda_i\rprE-(\rprJ-\rprR))=0. 
	\end{equation}
	The (right) residual direction is then given by $r_i = (\rprG-\rprP)^\T t_i$  and to enforce the interpolation conditions in \eqref{eq:OC1}, we set $\sigma_i=-\lambda_i$ and $b_i=r_i$ for all $i=1,\ldots,\,r$. This procedure is repeated, and upon convergence, the ROM satisfies the subset \eqref{eq:OC1} of $\htwo$ optimality conditions. Afterwards, we may attach the polynomial part ${\pHtf_{\rm pol}(s)=\prD + s\Dinf}$ to the strictly proper ROM as described in Section~\ref{subsec:compose_ROM}. One disadvantage of this approach, which is summarized in Algorithm \ref{alg:irkaPH}, is that a new ROM is computed in each iteration. In \cite{Moser2021}, an adaptation named \emph{CIRKA-PH} was proposed, which has the potential to significantly accelerate IRKA-PH, especially in large-scale settings for which interpolatory methods are particularly powerful. Embedding CIRKA-PH into the pH-DAE framework works in a similar way. 
	\begin{algorithm}[tpb]\caption{IRKA-PH for pH-DAEs (based on \cite{Gugercin2012})}\label{alg:irkaPH}
		\LinesNumbered
		\SetAlgoLined
		\DontPrintSemicolon
		\SetKwInOut{Input}{Input}\SetKwInOut{Output}{Output}
		\Input{Large-scale pH-DAE $\pHsys$ in staircase form; set of interpolation points $\{\sigma_1,\dots,\sigma_r\}$ and corresponding right tangential directions $\{b_1,\dots,b_r\}$ (both closed under complex conjugation).
		}
		\Output{Reduced pH-DAE $\pHsysr$. \;}			
		\While{not converged}{
			Perform steps 1-4 of Algorithm \ref{alg:arnoldi}. \;
			Compute $t_i\in\C^m, \lambda_i\in \C$ solving \eqref{eq:evp} for all $i=1,\dots, r$. \;
			$\sigma_i \gets - \lambda_i$ and $b_i \gets (\rprG-\rprP)^\T t_i \text{ for all}\; i=1,\dots, r$  \;			
		}
		Perform steps 5-10 of Algorithm \ref{alg:arnoldi}.\;
	\end{algorithm}	
	
	\subsection{Adaptive interpolation}
	In practice, besides the computational expense of IRKA-PH, it may sometimes be difficult to determine a suitable reduced order $r$ in the first place. In \cite{Druskin2011,Druskin2014}, an \emph{adaptive} approach was proposed that tackles this problem by iteratively adding new interpolation data in a complex region $\mathcal{S}$ where the approximation quality of the ROM is still poor. For pH-DAEs, this approach initially requires the computation of the proper system matrix $\prRos$, as described in Section \ref{subsec:trafo}. The approximation quality of a ROM generated in steps 1-4 of Algorithm \ref{alg:arnoldi} can then be assessed at points $\mu\in\mathcal{S}$ with the following residual matrix \cite{Druskin2014}:
	\begin{equation}
		\zeta(\mu) := (\prA-\mu\prE)\bar{V}_2(\rprJ-\rprR-\mu\rprE)^{-1}(\rprG-\rprP)-\prB.
	\end{equation}
	In each iteration, a new interpolation point $\sigma_{r+1}$ is added at the point in $\mathcal{S}$ where the norm of this residual matrix reaches its maximum, and a similar approach is taken to compute a new corresponding tangential direction $b_{r+1}$. Since $\sigma_{r+1}$ and $b_{r+1}$ are generally complex, their complex conjugates $\bar{\sigma}_{r+1}$ and $\bar{b}_{r+1}$ are also added to the interpolation data to keep it closed under conjugation. This way, the ROM dimension $r$ increases in each iteration until its transfer function does not significantly change between two subsequent iterations or the predefined maximum dimension $r_{\rm max}$ is reached. For strategies on how to choose $b_{r+1}$ and update the complex region $\mathcal{S}$, the interested reader is referred to \cite{Druskin2014}. The general approach is summarized in Algorithm \ref{alg:trksmPH}. 
	
	\begin{algorithm}[tpb]\caption{TRKSM-PH for pH-DAEs (based on \cite{Druskin2011,Druskin2014})}\label{alg:trksmPH}
		\LinesNumbered
		\SetAlgoLined
		\DontPrintSemicolon
		\SetKwInOut{Input}{Input}\SetKwInOut{Output}{Output}
		\Input{Large-scale pH-DAE $\pHsys$ in staircase form; set of interpolation points $\{\sigma_1,\dots,\sigma_r\}$ and corresponding right tangential directions $\{b_1,\dots,b_r\}$ (both closed under complex conjugation); maximum reduced order $r_{\rm max}>r_0$; initial complex region $\mathcal{S}$.
		}
		\Output{Reduced pH-DAE $\pHsysr$.\;}
		Compute and decompose $\prRos$ as in Section \ref{subsec:trafo}. \;			
		\While{not converged \textbf{and} $r < r_{max}$}{
			Perform steps 1-4 of Algorithm \ref{alg:arnoldi}.\;
			Solve \textbf{$\sigma_{r+1} = {\rm arg}\, \max_{\mu\in\mathcal{S}} \Vert \zeta(\mu)\Vert$}. \;
			Solve $b_{r+1} = {\rm arg}\, \max_{\Vert d\Vert=1} \Vert \zeta(\sigma_{r+1})d\Vert$. \;
			Add $(\sigma_{r+1},\bar{\sigma}_{r+1})$ and $(b_{r+1},\bar{b}_{r+1})$ to the interpolation data. \;
			Update the complex region $\mathcal{S}$. \;			
		}
		Perform steps 5-10 of Algorithm \ref{alg:arnoldi}.\;
	\end{algorithm}	
	
	\subsection{Interpolatory $\hinf$ approximation}
	So far, we have enforced that $\rprS=\prS$ and $\rprN=\prN$ to keep the $\htwo$ error bounded. For $\hinf$-inspired MOR, we only require $\pHtf-\pHtfr\in\mathcal{RH}_{\infty}^{m\times m}$ and thus, the reduced feedthrough matrices pose additional degrees of freedom that may be exploited in a similar manner as proposed in \cite{Beattie2009,Flagg2013,Castagnotto2017} for unstructured ODE systems. We may add structure-preserving perturbations to the feedthrough matrices $\rprN$ and $\rprS$ in the following way:
	\begin{alignat}{3}
		\perN &= \rprN + \Delta_N, \qquad &&\Delta_N &&= -\Delta_N^\T, \\
		\perS &= \rprS + \Delta_S, \qquad &&\Delta_S &&= \Delta_S^\T  \geq 0.	
	\end{alignat} 
	Simply adding these perturbations would only change the direct feedthrough of the ROM but not the dynamics and is therefore not expected to improve the $\hinf$ approximation quality significantly. However, as shown in \cite{Beattie2009}, it is also possible to perturb the feedthrough matrix of the ROM while retaining predefined interpolation conditions with the FOM --- and the same holds for pH-DAEs.
	\begin{lemma}\label{lem:ihaPH}
		Assume that we have obtained a reduced system matrix ${\rRos=s\mathcal{E}_r+\Gamma_r+W_r}$ in steps 1-4 of Algorithm \ref{alg:arnoldi} using a set of interpolation points $\{\sigma_1,\dots,\sigma_r\}$ and corresponding right tangential directions $\{b_1,\dots,b_r\}$, both closed under complex conjugation. Let $F\in\R^{n_2\times m}$ be a solution to
		\begin{equation*}
			F^\T\bar{V}_2 = [b_1,\ldots,b_r]T_v,
		\end{equation*}
		with $T_v\in\C^{r\times r}$ such that $\bar{V}_2 = V_2T_v$. If the system matrix $\rRos$ is perturbed such that
		\begin{equation*}
			\widehat{\Ros}_r = s\mathcal{E}_r + \widehat{\Gamma}_r + \widehat{W}_r,
		\end{equation*}		
		with
		\begin{alignat}{4}
			\widehat{\Gamma}_r &= \mat{-\perJ & -\perG \\ \perG^\T & \perN} &&= \Gamma_r &&+ \mat{\bar{V}_2^\T F\\-I_m}\Delta_{\pN}&&\mat{\bar{V}_2^\T F\\-I_m}^\T,\\
			\widehat{W}_r&=\mat{\perR & \perP \\ \perP^\T & \perS} &&= W_r &&+  \mat{\bar{V}_2^\T F\\-I_m}\Delta_{\pS}&&\mat{\bar{V}_2^\T F\\-I_m}^\T,
		\end{alignat}
		then the perturbed ROM $\widehat{\pHsys}_r$ with transfer function $\widehat{\pHtf}_r$ obtained by steps 5-10 in Algorithm \ref{alg:arnoldi} is a pH-DAE system, and it holds that
		\begin{equation}\label{eq:mom_match_ihaPH}
			\pHtf(\sigma_i)b_i = \pHtfr(\sigma_i)b_i = \widehat{\pHtf}_r(\sigma_i)b_i,
		\end{equation}
		for all $i=1,\ldots,r$ and for any $\Delta_N = -\Delta_N^\T$ and $\Delta_S = \Delta_S^\T \geq 0$.
	\end{lemma}
	\begin{proof}
		The fact that the perturbed system $\widehat{\pHsys}_r$ fulfils the pH structural constraints follows directly from the properties of $\Delta_N$ and $\Delta_S$. The proof of \eqref{eq:mom_match_ihaPH} follows the proof of Theorem 3 in \cite{Beattie2009} for general LTI systems and is therefore omitted here.  
	\end{proof}
	This result enables us to optimize the new degrees of freedom $\Delta_N, \Delta_S$ in an $\hinf$-inspired way. Optimizing the $\hinf$ error directly is challenging since $\hinf$ norm computations are computationally taxing, and the $\hinf$ norm depends non-smoothly on $\Delta_N, \Delta_S$ (see \cite[Section 3.2.1]{Schwerdtner2020}). Instead, the SOBMOR algorithm, proposed in \cite{Schwerdtner2020}, may be employed. Therein, the functions $\vtu(\cdot)$ and $\vtsu(\cdot)$ are introduced, which map vectors row-wise to appropriately sized upper triangular and strictly upper triangular matrices, respectively. The function names are abbreviations for \emph{vector-to-upper} and \emph{vector-to-strictly-upper}, respectively. Using these functions, we may define parameter vectors $\theta_N\in\R^{m(m-1)/2}$ and $\theta_S\in\R^{m(m+1)/2}$ and design $\Delta_N, \Delta_S$ in the following way:
	\begin{align*}
		\Delta_N(\theta_N) &:= \vtsu(\theta_N)^\T - \vtsu(\theta_N), \\
		\Delta_S(\theta_S) &:= \vtu(\theta_S)^\T\vtu(\theta_S). \\		
	\end{align*}
	Finally, a levelled least-squares approach can be taken to optimize the error ${\Vert\pHtf-\widehat{\pHtf}_r(\cdot,\theta)\Vert_{\hinf}}$ with ${\theta^\T:=[\theta_N^\T,\theta_S^\T]^\T\in\R^{\dimu^2}}$, as described in \cite{Schwerdtner2020,Schwerdtner2021Ident}, and to which we refer for implementation details. We summarize the approach in Algorithm \ref{alg:ihaPH}.
	
	\subsection{Suitable representations for MOR}\label{subsec:kyp}
	Since we apply Galerkin projections to preserve the pH structure, transformations of the FOM under s.s.e. will not change its transfer function, but they will have an impact on MOR, which raises the question of how to find suitable representations of the FOM that yield better approximations. This was recently examined in \cite{BreU21} for explicit ODE systems, and may be incorporated into our framework as follows. 

	\begin{algorithm}[tpb]\caption{IHA-PH for pH-DAEs (based on \cite{Flagg2013})}\label{alg:ihaPH}
		\LinesNumbered
		\SetAlgoLined
		\DontPrintSemicolon
		\SetKwInOut{Input}{Input}\SetKwInOut{Output}{Output}
		\Input{Large-scale pH-DAE $\pHsys$ in staircase form; set of interpolation points $\{\sigma_1,\dots,\sigma_r\}$ and corresponding right tangential directions $\{b_1,\dots,b_r\}$ (both closed under complex conjugation); initial parameter vector $\theta\in\R^{\dimu^2}$.
		}
		\Output{Perturbed reduced pH-DAE $\widehat{\pHsys}_r$\;}
		Compute the reduced system matrix $\rRos$ with steps 1-5 in Algorithm \ref{alg:irkaPH}. \;			
		Solve \textbf{$\theta^* = {\rm arg}\, \min_{\theta\in\R^{m^2}} \Vert\pHtf-\widehat{\pHtf}_r(\cdot,\theta)\Vert_{\hinf}$} using the approach in \cite{Schwerdtner2020,Schwerdtner2021Ident} \;
		Compute $\widehat{\Ros}_r(\theta^*)$ as in Lemma \ref{lem:ihaPH}.\;	
		Construct $\widehat{\pHsys}_r(\theta^*)$ with steps 5-10 in Algorithm \ref{alg:arnoldi}.\;
	\end{algorithm}		

	Assume that we have computed the proper system matrix
	\begin{equation*}
		\prRos(s) = \mat{s\prE-\prA&-\prB\\\prC&\prD}, 
	\end{equation*}
	of the FOM as in Section \ref{subsec:trafo}. In Lemma \ref{lem:prop_decomp}, we obtained a pH representation of this subsystem by simply decomposing the system matrix into symmetric and skew-symmetric parts. However, there are other ways, and the family of pH representations for this system is parameterized by the \emph{Kalman-Yakubovich-Popov (KYP) inequality}, as shown in \cite{Beattie2019}. If the proper system is \emph{behaviourally observable}, i.e., ${\text{rank}[(s\prE - \prA)^\T, \prC^\T] = n_2}$ for all $s \in \C$, which we assume in the following, the KYP inequality	
	\begin{equation}\label{eq:KYP}
		\mat{-\prA^\T X - X^\T \prA & \prC^\T - X^\T \prB\\ \prC - \prB^\T X & \prD + \prD^\T}\geq 0,\quad X^\T\prE = \prE^\T X \geq 0,
	\end{equation}
	has symmetric solutions $X\in\R^{n\times n}$ that are bounded such that
	\begin{equation*}
		0 < X_- \leq X \leq X_+
	\end{equation*}
	with minimal and maximal solutions $X_-$ and $X_+$, respectively (see \cite[Theorem 1]{Beattie2019}). Here, $X_-\leq X$ reflects that $X - X_-$ is positive semi-definite. Now assume that we apply another transformation under s.s.e. on $\prRos$ using any $X$ in the following way
	\begin{equation*}
		\tRos_{\rm p} = \mat{X^\T & 0 \\ 0 & I_m} \prRos \mat{I_{n_2} & 0 \\ 0 & I_m},
	\end{equation*}
	then we obtain the decomposition
	\begin{equation*}
		\tRos_{\rm p} = s\widetilde{\mathcal{E}}_{\rm p}+\tGamma_{\rm p}+\tW_{\rm p},
	\end{equation*}
	with matrices 
	\begin{alignat}{1}
		\widetilde{\mathcal{E}}_{\rm p} &= \mat{X^\T \prE & 0 \\ 0 & 0}, \\
		\tGamma_{\rm p} &= \mat{\prA^\T X - X^\T \prA & -(\prC^\T + X^\T \prB)\\ \prC + \prB^\T X & \prD - \prD^\T},\\
		\tW_{\rm p} &= \mat{-\prA^\T X - X^\T \prA & \prC^\T - X^\T \prB\\ \prC - \prB^\T X & \prD + \prD^\T},
	\end{alignat}	
	which clearly fulfils the pH structural constraints due to \eqref{eq:KYP}. The Hamiltonian of the system associated with $\tRos_{\rm p}$ changes to
	\begin{equation*}
		\tra{\mathcal{H}}_{\rm p}(\tra{x}_{\rm p}) = \frac{1}{2}\tra{x}_{\rm p}^\T X^\T\prE \tra{x}_{\rm p}.
	\end{equation*}
	where $\tra{x}_{\rm p}$ denotes the new state vector of the transformed proper subsystem. As discussed in \cite{BreU21}, choosing the minimal solution $X_-$ for the transformation is particularly suited for MOR and may significantly improve the approximation quality. In our framework, we can include this transformation change by simply replacing $\We$ in step 3 of Algorithm \ref{alg:arnoldi} by
	\begin{equation*}
		\We_- = \mat{0&A_{14}^{-\T}C_4^\T\\ X_-\bar{V}_2& 0\\-A_{33}^{-\T}A_{23}^\T X_-\bar{V}_2 & 	A_{33}^{-\T}(C_3^\T-A_{13}^\T A_{14}^{-\T} C_4^\T)\\0&0\\0&I_m}.
	\end{equation*}	
	Note that this does not affect the tangential interpolation conditions since we retain $\Ve$ and also does not affect the matrices $\rprS$, $\rprN$ or $\Dinf$ since we keep the second block column of $\We$ unchanged. However, it does indeed have an effect on the strictly proper part $\pHtf_{\rm sp,r}$ of the ROM's transfer function \emph{between} the interpolation points, which is illustrated in the next section by numerical examples. Note that, especially if the matrices $\prE, \prA$ are very large or even dense (see Example \ref{ex:sparsity}), the solution of \eqref{eq:KYP} may be computationally taxing. Therefore, efficient low-rank and/or sparse approximations of $X_-$ are required that ideally originate from the sparse FOM matrices, which is an open research problem. 
	
	\section{Numerical experiments} \label{sec:num_examples}
	A practical example where pH-DAE models naturally occur is the modelling of electrical RCL networks consisting of linear resistors, capacitors and inductors, which are used, for example, in the simulation of VLSI circuits or transmission lines. RCL networks are typically modeled using modified nodal analysis (MNA), which is also used by simulation software such as \emph{SPICE} (see, e.g., \cite{Freund2011}). When large-scale RCL networks are reduced, it is crucial that the ROM retains the passivity property of the original model to couple the ROM with other (possibly non-linear) parts of the system. Interpolatory reduction methods that preserve the passivity and MNA structure of the original model are given by the PRIMA \cite{Odabasioglu1998} and SPRIM \cite{Freund2011} algorithms, respectively.
	Since RCL models may be transformed into pH-DAE staircase form, and since pH-DAE models are inherently passive, we may also preserve the passivity property with the algorithms presented in this work.
	
	We consider two RCL ladder networks \texttt{RCL-1M}\footnote{The FOM system matrices of \texttt{RCL-1M} are available at \url{https://doi.org/10.5281/zenodo.6497076}} and \texttt{RCL-12S}\footnote{The FOM system matrices of \texttt{RCL-12S} are available at \url{https://doi.org/10.5281/zenodo.6602125}}, which were also used in \cite{MosSMV22} and \cite{MosSMV2022b}, respectively, and are generated with the port-Hamiltonian benchmark collection\footnote{\url{https://port-hamiltonian.io}} to which we refer for a more detailed physical description. The model specifications and sigma plots are given in Figure \ref{fig:foms}. \texttt{RCL-1M} is a large-scale model with two input-output pairs and differentiation index 1, which generates an approximately constant input-output gain for higher frequencies. Model \texttt{RCL-12S} has significantly smaller dimensions and only one input and output but contains index-1 and index-2 parts, which leads to an improper transfer function with $D_\infty\neq 0$. 
	
	Let us first consider the reduction of \texttt{RCL-1M}. We reduce the original model to dimension $r=40$ using the algorithms IRKA-PH, TRKSM-PH and IHA-PH. The sigma plots of the resulting ROM transfer functions and respective error systems $\pHtf-\pHtfr$ are plotted in Figure \ref{fig:res_RCL-1M}. For this example, the proper matrices $\prE,\,\prA$ are sparse, and therefore, the residual $\zeta$ in TRKSM-PH can be evaluated very efficiently. This makes TRKSM-PH a computationally efficient alternative to IRKA-PH since it yields a comparable performance while requiring significantly fewer solutions to large-scale linear systems of equations. The additional degrees of freedom in IHA-PH, on the other hand, enable more accurate results in small frequency regions at the expense of a constant error gain at higher frequencies, which results from the perturbation of the reduced feedthrough matrix.
	
	For model \texttt{RCL-12S}, we compute ROMs of different dimensions ranging from ${r=2}$ to ${r=20}$. The $\htwo$ and $\hinf$ errors are plotted in Figure \ref{fig:res_RCL-12S} for each MOR method presented in Section \ref{sec:extensions}. For all methods, the errors decrease for increasing reduced orders $r$, which is expected since more interpolation conditions can be enforced. TRKSM-PH again yields similar $\htwo$ errors as IRKA-PH for larger ROM dimensions. For IHA-PH, only the $\hinf$ errors are plotted in Figure \ref{fig:res_RCL-12S} since it produces unbounded $\htwo$ errors due to the perturbation of the feedthrough matrix. For most reduced orders, the $\hinf$ errors are only marginally smaller than those produced by IRKA-PH since the model has only one input-output pair and consequently, $\theta\in\R$. However, as shown for $r=18$, even one additional parameter may improve the $\hinf$ approximation quality. Finally, we also illustrate the importance of choosing a suitable reduction matrix $\We$. Replacing the matrix $\We$ in IRKA-PH by $\We_-$, as described in Section \ref{subsec:kyp}, yields significantly smaller errors both in the $\htwo$ and $\hinf$ norm. Note that a similar basis change could, of course, also be applied to the IHA-PH and TRKSM-PH methods which is expected to yield similar improvements but is omitted here.
	
\begin{figure}[tpb]%
	\centering
	\begin{tabular}{cc}
		\includegraphics[width=0.45\textwidth]{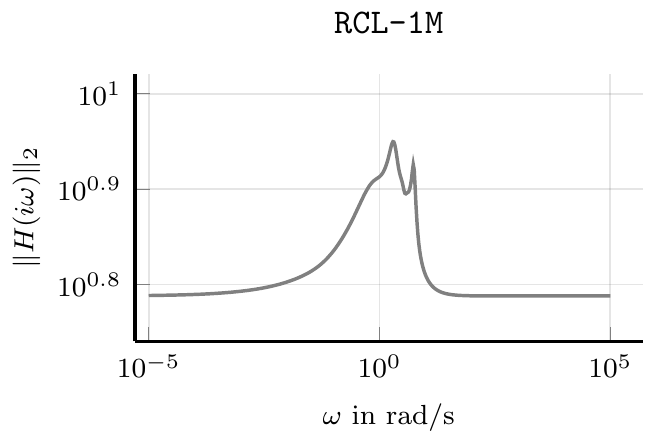} & \includegraphics[width=0.45\textwidth]{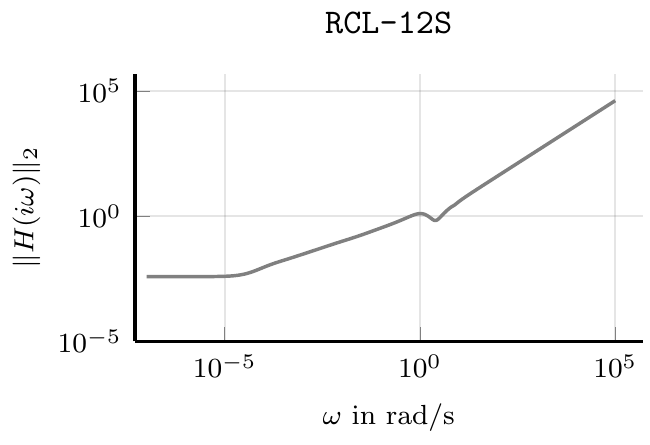}
	\end{tabular}	
		\subfloat{
		\begin{tabular}{ |c||c|c|c|c|c|c|c| } 
			\hline
			Model name & $\nu$ & $n_2$ & $n_3$ & $n_4/n_1$ & $\Dinf$ & $n$ & $m$ \\
			\hline 
			\hline
			\texttt{RCL-1M} & $1$ & $19\,999$ & $10\,005$ & $0$ & $0$ & $30\,004$ & $2$\\
			\hline 
			\texttt{RCL-12S} & $2$ & $999$& $501$ & $1$ & $\neq 0$ & $1\,502$ & $1$\\
			\hline 
		\end{tabular}
	} 
	\caption{Model parameters and sigma plots for the RCL ladder network models \texttt{RCL-1M} and \texttt{RCL-12S} in pH-DAE staircase form.}
	\label{fig:foms}
\end{figure}

\begin{figure}[tbh]%
	\centering
	\includegraphics[width=0.9\textwidth]{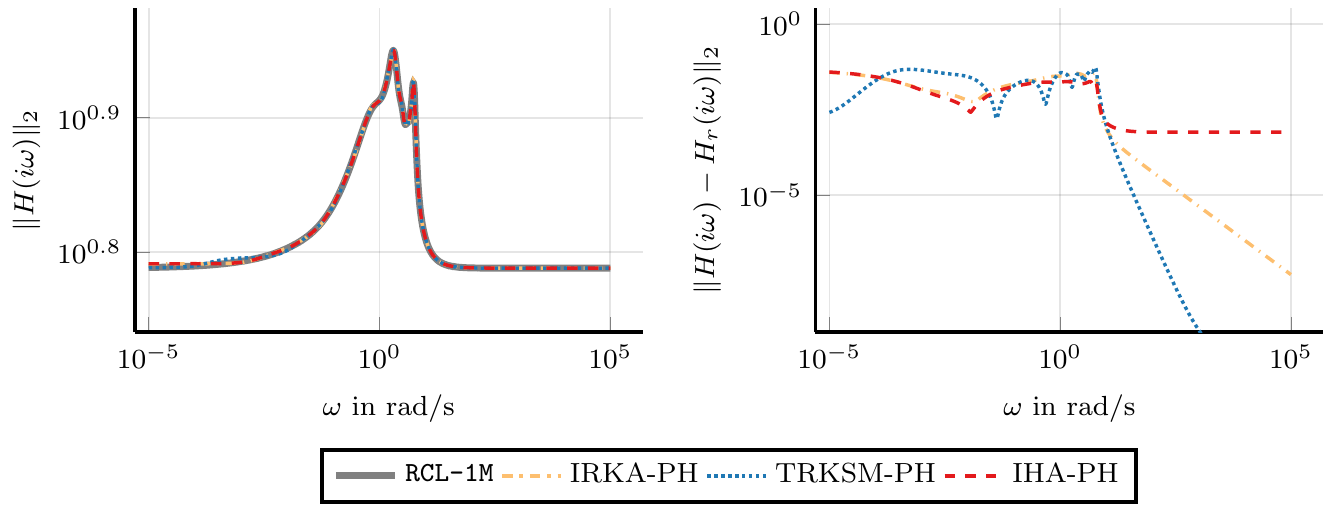}
	\caption{Reduction of the model \texttt{RCL-1M} to order ${r=40}$ using different interpolatory MOR methods. Given are the sigma plots of the FOM and ROMs (left) and of the respective error systems (right).}
	\label{fig:res_RCL-1M}
\end{figure}
\begin{figure}[tbh]%
	\centering
	\includegraphics[width=0.9\textwidth]{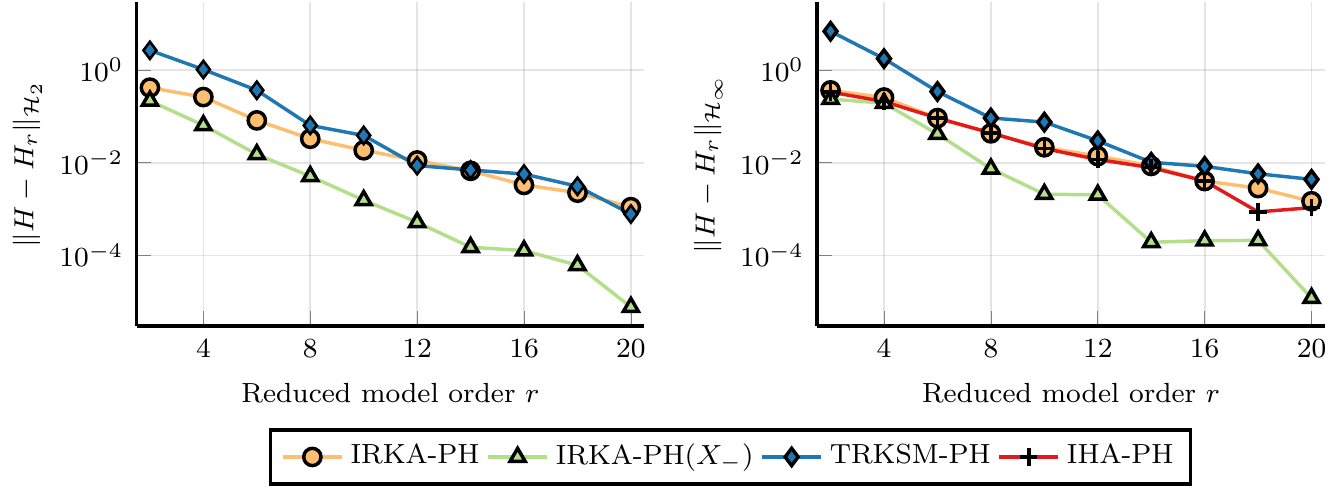}
	
	\caption{Reduction of the model \texttt{RCL-12S} to different reduced orders $r\in\{2,4,\ldots,20\}$. Plotted are the $\htwo$ errors (left) and $\hinf$ errors (right) for different interpolatory MOR methods.}
	\label{fig:res_RCL-12S}
\end{figure}

	\section{Conclusion} \label{sec:conclusion}
	The Rosenbrock system matrix exhibits a particular structure for pH-DAE systems that can be exploited for model reduction. We have deduced a novel interpolatory MOR framework for pH-DAEs in staircase form, which, compared to other tangential interpolation methods, guarantees minimal ROMs in pH-DAE staircase form irrespective of the original system's differentiation index. Moreover, its simple structure allows the incorporation of different strategies for choosing suitable interpolation data which were originally proposed for unstructured DAE systems. In electrical circuit simulation where pH-DAE models naturally arise, our framework can be considered as an alternative to traditional, passivity-preserving MOR methods, which we illustrated using two numerical examples.

	\section*{Acknowledgements}
	This research has been funded by the Deutsche Forschungsgemeinschaft (DFG, German Research Foundation) - Project number 418612884. We also sincerely appreciate the help of Maximilian Bonauer, Nora Reinbold and Paul Schwerdtner with the implementation of the algorithms in Section \ref{sec:extensions}.
	
	\bibliography{references}
	
\end{document}